\newcommand*{\QUANTUM}{}%
\ifdefined\QUANTUM
\documentclass[a4paper,onecolumn,11pt,accepted=2023-10-28]{quantumarticle}
\pdfoutput=1
\else
\documentclass[11pt]{article}
\fi

\usepackage{amsmath,amssymb,amsthm}
\usepackage{latexsym}
\usepackage{indentfirst}
\usepackage{placeins}
\usepackage{booktabs}
\usepackage{algorithm}
\usepackage{algorithmic}
\usepackage{multirow}
\usepackage{color, xcolor, colortbl}
\usepackage{subfigure}
\usepackage{qcircuit}

\usepackage{graphicx}
\usepackage{bm}
\usepackage{braket}

\usepackage[sort,square,numbers]{natbib}
\usepackage[T1]{fontenc}
\usepackage{hyperref}
\usepackage[capitalize]{cleveref}
\makeatletter

\theoremstyle{plain}
\newtheorem{theorem}{Theorem}
\newtheorem{lemma}[theorem]{Lemma}

\newtheorem{rmk}[theorem]{Remark}

\theoremstyle{definition}

\theoremstyle{definition}

\theoremstyle{remark}

\newcommand{\abs}[1]{\left|#1\right|}

\newcommand{\eps}{\epsilon}
\newcommand{\CC}{\mathbb{C}}
\newcommand{\PP}{\mathbb{P}}

\renewcommand{\tt}{T_{\mathrm{total}}}
\newcommand{\tm}{T_{\mathrm{max}}}

\newcommand{\mo}[1]{\mathcal{O}\left(#1\right)}
\newcommand{\tmo}[1]{\tilde{\mathcal{O}}\left(#1\right)}
\newcommand{\ntp}[1]{\left|#1\right|_{2\pi}}

\begin{document}

\title{On low-depth algorithms for quantum phase estimation}

\author{Hongkang Ni} 
\affiliation{Institute for Computational and Mathematical Engineering, Stanford University,  Stanford, CA 94305} 
\orcid{0000-0002-7507-4755}
\email{hongkang@stanford.edu}

\author{Haoya Li}
\affiliation{Department of Mathematics, Stanford University, Stanford, CA 94305} 
\orcid{0000-0001-7076-7600}
\email{lihaoya@stanford.edu}

\author{Lexing Ying} 
\affiliation{Department of Mathematics, Stanford University, Stanford, CA 94305}
\affiliation{Institute for Computational and Mathematical Engineering, Stanford University,  Stanford, CA 94305} 
\orcid{0000-0003-1547-1457}
\email{lexing@stanford.edu}

\thanks{We thank Lin Lin for communicating the recent development of quantum phase estimation. The work of L.Y. is partially supported by the National Science Foundation under awards DMS-2011699 and DMS-2208163}

\begin{abstract}
Quantum phase estimation is one of the critical building blocks of quantum computing. For early fault-tolerant quantum devices, it is desirable for a quantum phase estimation algorithm to (1) use a minimal number of ancilla qubits, (2) allow for inexact initial states with a significant mismatch, (3) achieve the Heisenberg limit for the total resource used, and (4) have a diminishing prefactor for the maximum circuit length when the overlap between the initial state and the target state approaches one. In this paper, we prove that an existing algorithm from quantum metrology can achieve the first three requirements. As a second contribution, we propose a modified version of the algorithm that also meets the fourth requirement, which makes it particularly attractive for early fault-tolerant quantum devices. 
\end{abstract}


\maketitle

\section{Introduction}\label{sec:intro}
Quantum phase estimation (QPE) is one of the most important building blocks of quantum computing. Consider a unitary matrix $U\in\mathbb{C}^{M\times M}$. Let $\{\ket{\psi_m}\}_{m=0}^{M-1}$ be the orthogonal eigenstates of $U$, and $\{e^{\mathrm{i}\lambda_m}\}_{m=0}^{M-1}$ be the corresponding eigenvalues. In the QPE problem, given $U$ and an eigenstate, say $\ket{\psi_0}$, the goal is to estimate $\lambda_0$ up to a certain accuracy. In a more general and practical setting, the initial quantum state $\ket{\psi}$ provided for phase estimation is a linear combination of eigenstates, i.e., $\ket{\psi}=\sum_{m=0}^{M-1}c_m\ket{\psi_m}$, where the coefficient $c_0$ for $\ket{\psi_0}$ dominates.

Due to the significance of QPE, many algorithms have been devised to address this problem. The well-known Hadamard test provides probably the simplest circuit (see \Cref{fig:circuits}(a)) for this purpose, but it requires $\mo{\eps^{-2}}$ executions to reach a precision $\eps$. Improving on the Hadamard test, Kitaev's algorithm \cite{kitaev1995quantum,kitaev2002classical} identifies the phase bit-by-bit by using quantum circuits with dyadic powers $U^{2^j}$ (see \Cref{fig:circuits}(b)). It reduces the total complexity but is only applicable to exact eigenstates. The QPE algorithm based on quantum Fourier transform (QFT) \cite{cleve1998quantum} requires only a single run but a relatively large number of ancilla qubits. Many alternatives for QPE have been proposed in the literature \cite{berry2015simulating,higgins2007entanglement,knill2007optimal,poulin2009sampling,o2019quantum,dong2022ground,lin2020near,lin2022heisenberg,ding2022even,zhang2022computing,wang2022quantum,wan2022randomized}. Among them, \cite{zhang2022computing,wang2022quantum} focus on the case of ground state energy estimation and a spectral gap is assumed. For a more comprehensive overview of the QPE algorithms, we refer to the detailed discussions in \cite{ding2022even,lin2022heisenberg,nielsen2001quantum}.

\begin{figure}
  \centering
  \begin{tabular}{cc}
    \includegraphics[scale=0.4]{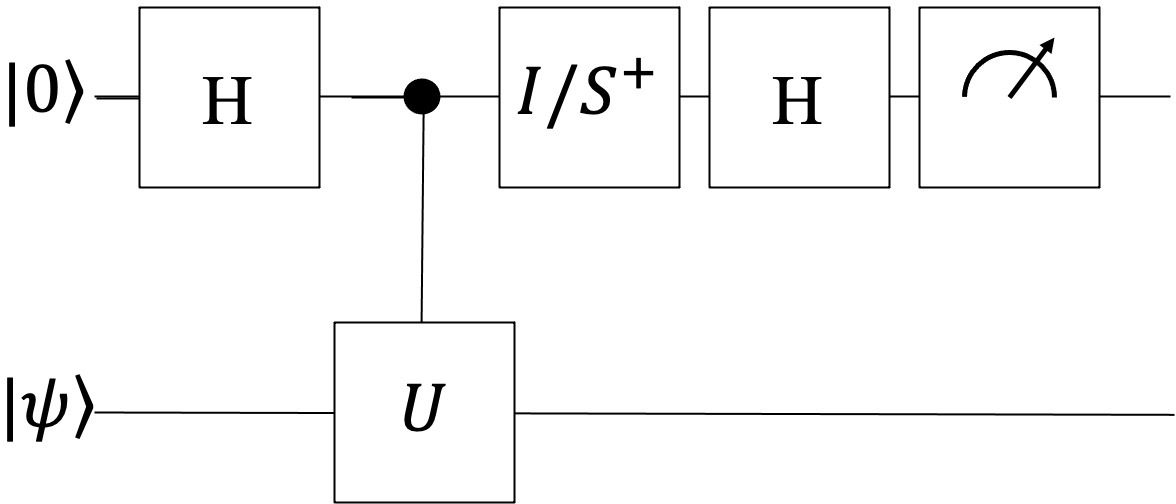} &
    \includegraphics[scale=0.4]{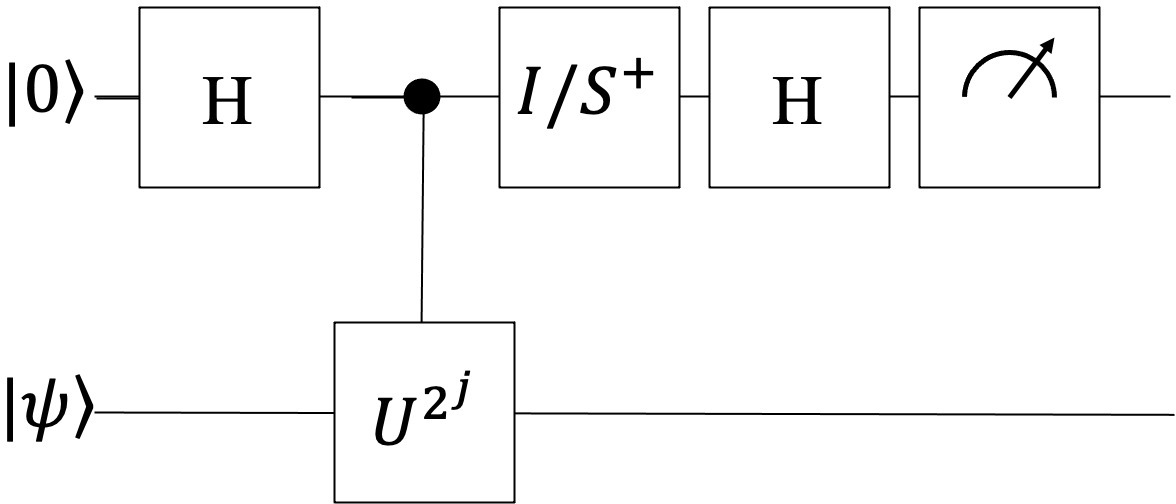}\\
    (a) & (b)
  \end{tabular}
  \caption{(a) The circuit for Hadamard test. $\mathrm{H}$ is the Hadamard gate. Concerning the $I/S^+$ gate, we use $I$ (the identity) for the real part of $\braket{\psi|U|\psi}$, and $S^+$ (the adjoint of the phase gate $S$) for the imaginary part. (b) The circuit used in the Kitaev algorithm estimates the real and imaginary parts of $\braket{\psi|U^{2^j}|\psi}$ for multiple integer values of $j$.  }
  \label{fig:circuits}
\end{figure}

There are several key complexity metrics for evaluating the performance of the QPE algorithms. The first is the number of ancilla qubits required; the smaller, the better. The second one is the maximum runtime $\tm$, measured by the maximum depth of any circuit used by the algorithm. The third one is the total runtime $\tt$, equal to the sum of the circuit depths over all executions. It has been demonstrated (see \cite{giovannetti2006quantum, zwierz2010general, zhou2018achieving} for example) that $\tt$ must follow the Heisenberg limit $\tt=\Omega(\eps^{-1})$. The fourth one is the minimal overlap $p_0 = |c_0|^2$ required between the initial state $\ket{\psi}$ and the target state $\ket{\psi_0}$. A lower bound of this overlap is usually assumed because the problem of finding $\lambda_0$ is otherwise shown to be difficult (\cite{kitaev2002classical, kempe2006complexity, aharonov2002quantum}). Among these metrics, a small $\tm$ is particularly important for early fault-tolerant quantum devices since these devices typically have a small number of qubits and a relatively short coherence time. 
For the phase estimation problem, one can use similar ideas to the proof of lower bound in \cite{huang2023learning} to show that $\tt = \Omega(\eps^{-2}\tm)$, which means $\tm=\Omega(\eps^{-1})$ when $\tt = \mo{\eps^{-1}}$. A detailed proof is given in \cite{yu2021}. In conclusion, a near-optimal phase estimation algorithm should meet the following requirements:
\begin{enumerate}
\item Use a small number of (even a single) ancilla qubits. \label{1}
\item Allow the initial state $\ket{\psi}$ to be inexact. \label{2}
\item Achieve the Heisenberg-limited scaling $\tt = \tmo{\eps^{-1}}$. \label{3}
\item $\tm = \mo{\eps^{-1}}$, and the prefactor can be arbitrarily small when $\ket{\psi}$ approaches to the exact eigenstate $\ket{\psi_0}$. As we mentioned, this is particularly important for early fault-tolerant quantum devices. \label{4}
\end{enumerate}
Most of the proposed QPE algorithms fail to meet all four requirements. For example, the first requirement is violated by QPE algorithms using QFT. The Hadamard test and the original Kitaev algorithm violate the second requirement. The Hadamard test also violates the third requirement. It has been pointed out in \cite{ding2022even} that the only algorithm proven to meet all four conditions is the recent optimization-based algorithm proposed in \cite{ding2022even}.

There have also been rapid developments of phase estimation in the related field of quantum metrology, such as \cite{kimmel2015robust,belliardo2020achieving,lumino2018experimental,rudinger2017experimental,russo2021evaluating}. For example, the robust phase estimation (RPE) algorithm in \cite{kimmel2015robust,belliardo2020achieving,russo2021evaluating} halves the confidence interval of the phase iteratively to achieve the target accuracy. In this paper,
\begin{itemize}
\item we show that this RPE algorithm satisfies the first three requirements listed above as long as the initial overlap is above $4-2\sqrt{3}\approx 0.536$. This is an improvement over the threshold $0.71$ obtained in \cite{ding2022even}, and
\item we propose a modified algorithm with a much shorter circuit length when the overlap approaches $1$. The prefactor of $\tm = \tmo{\eps^{-1}}$ can be as small as $\Theta(1-p_0)$, which is better than the bound $\Theta(\sqrt{1-p_0})$ provided in \cite{ding2022even}.
\end{itemize}
The rest of the paper is organized as follows. \Cref{sec:alg} proves the correctness of RPE when the initial overlap is above $4-2\sqrt{3}$. \Cref{sec:new} presents the new algorithm that allows for shorter circuit length when the initial overlap approaches one.

\section{Analysis of the existing algorithm}\label{sec:alg}

The angles and their computations are understood as modulo $2\pi$. The absolute value of an angle $\theta$, denoted by $\ntp{\theta}$, is defined to be the minimum distance to $0$ modulo $2\pi$, i.e., $\ntp{\theta} = \pi - |\theta \mod2\pi-\pi|$.  For any quantum state $\ket{\psi}$, denote by $p_m = |c_m|^2 = \abs{\braket{\psi|\psi_m}}^2$ the overlap between the given state and the eigenstate $\ket{\psi_m}$.

In the Hadamard test, the circuit in \Cref{fig:circuits}(a) is used to estimate the real and imaginary parts of $\braket{\psi|U|\psi}$. In Kitaev's algorithm, the circuits in \Cref{fig:circuits}(b) are used to estimate the real and imaginary parts of $\braket{\psi|U^{2^j}|\psi}$ for different values of $j$.
\Cref{alg:pruning} is a reformulation of the RPE outlined in \cite{belliardo2020achieving}, which uses the same primitives of the Kitaev's algorithm and produces an approximation of $\lambda_0$ with high probability. In the $j$-th step, the Hadamard test is executed $N_s$ times to obtain an estimate $Z_j$ of $2^j\lambda_0$. Then $\arg Z_j$ is used to get a candidate set $S_j$ of new approximations of $\lambda_0$. The one closest to the previous approximation $\theta_{j-1}$ is chosen as the new approximation $\theta_j$.

\begin{algorithm}[ht]
	\caption{An adapted version of RPE in \cite{belliardo2020achieving}}
	\label{alg:pruning}
	\begin{algorithmic}
		\STATE{\textbf{Input:} $\eps$: target accuracy, $\eta$: upper bound of the failure probability, $\delta$: upper bound for the noise in the initial state $\ket{\psi}$. }
		\STATE{Let $J = \lceil\log_2(\eps^{-1})\rceil$ and calculate $N_s$ with the values of $\eps$, $\eta$ and $\delta$ according to \eqref{eq:Ns}.}
		\STATE{$\theta_{-1} = 0$.}
		\FOR{$j = 0,1,\ldots,J$}
		\STATE{Run the circuit in \Cref{fig:circuits}(b) for the real part and imaginary part of $\braket{\psi|U^{2^j}|\psi}$ for $\frac{N_s}{2}$ times each to generate $Z_j$ as an estimation of $\braket{\psi|U^{2^j}|\psi}$.}
		\STATE{Define a candidate set $S_j = \left\{\frac{2k\pi+\arg Z_j}{2^j} \right\}_{k=0,\ldots,2^j-1}$.}
		\STATE{$\theta_j = \arg\min_{\theta\in S_j}\ntp{\theta - \theta_{j-1}}$. }	
		\ENDFOR
		\STATE{\textbf{Output:} $\theta_J$ as an approximation to $\lambda_0$. }
	\end{algorithmic}
\end{algorithm}
The following lemma is key to the analysis of \Cref{alg:pruning}. Here, the constant $\delta$ serves as an upper bound for the noise in the initial state, i.e., $\delta > p_1 + \ldots p_{M-1} $.
\begin{lemma}\label{lem:1}
	Suppose the constant $\delta < 2\sqrt{3}-3$ and let 
	\begin{equation}\label{eq:alpha}
	  \alpha(\delta) = \frac{\sqrt{3}}{2}(1-\delta) - \delta >0.
	\end{equation} 
	If the quantum state $\ket{\psi}$ satisfies $p_0 > 1-\delta$ and
	\begin{equation}
		\abs{Z_j - \braket{\psi|U^{2^j}|\psi}} < \alpha(\delta),
	\end{equation} 
	then 
	\begin{equation}
		2^j\lambda_0 \in \left(\arg Z_j-\frac{\pi}{3}, \arg Z_j+\frac{\pi}{3}\right)\mod 2\pi,\label{eq:range_j}
	\end{equation}
	where $\arg Z_j$ is the principal argument of $Z_j$.
\end{lemma}

\begin{proof}
	Direct calculation shows
	\begin{equation}
		\begin{aligned}
		  \alpha(\delta) &> \abs{Z_j - \braket{\psi|U^{2^j }|\psi}}	= \abs{Z_j - p_0 e^{\mathrm{i} 2^j \lambda_0} - \sum_{m = 1}^M p_m e^{\mathrm{i} 2^j \lambda_m}}\\
			&\ge \abs{Z_j - p_0 e^{\mathrm{i} 2^j \lambda_0}} - \sum_{m = 1}^M p_m  \ge \abs{Z_j - p_0 e^{\mathrm{i} 2^j \lambda_0}} - \delta,
		\end{aligned}
	\end{equation}
	which means $Z_j$ must be in a ball $B_{\alpha(\delta)+\delta}(p_0 e^{\mathrm{i} 2^j \lambda_0})  \subset\CC$. Noticing that $\alpha(\delta)+\delta = \frac{\sqrt{3}}{2}(1-\delta)$ and $p_0>1-\delta$, the sine of the angle between $Z_j$ and $e^{\mathrm{i} 2^j \lambda_0}$ is bounded by $\frac{\sqrt{3}}{2}$ and hence \eqref{eq:range_j} holds (see \Cref{fig:lemma}).
\end{proof}

\begin{figure}[ht!]
  \centering
  \includegraphics[scale=0.35]{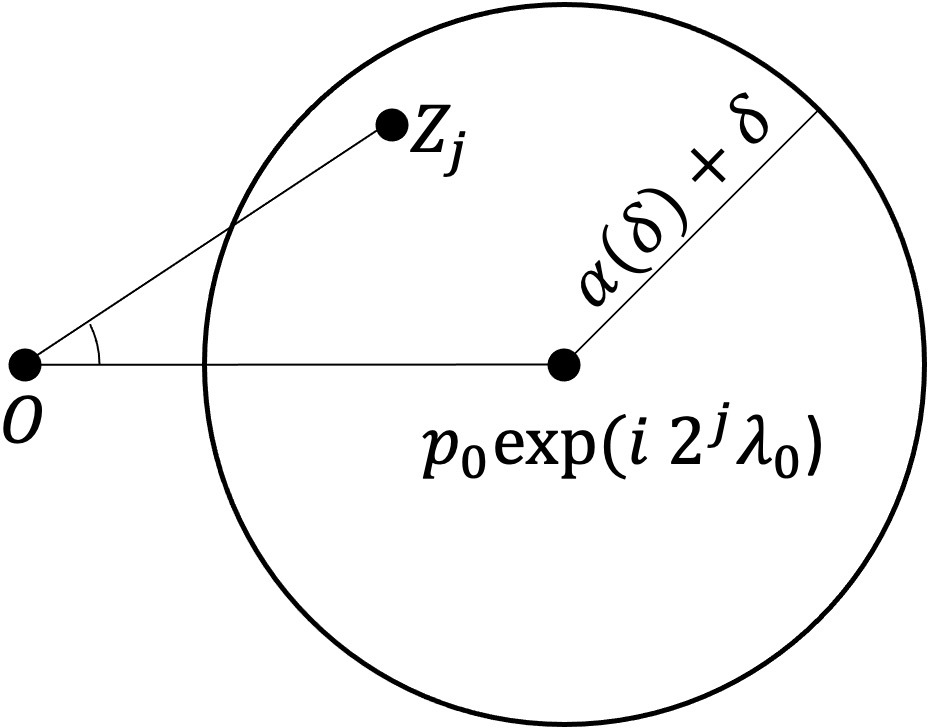}
  \caption{Illustration of the proof of \Cref{lem:1}.}
  \label{fig:lemma}
\end{figure}

The following theorem is the main theoretical guarantee of \Cref{alg:pruning}.

\begin{theorem}\label{thm:1}
Suppose the constant $\delta < 2\sqrt{3}-3$ and that the quantum state $\ket{\psi}$ satisfies $p_0 > 1-\delta > 4-2\sqrt{3} \approx 0.536$ and 
	\begin{equation}\label{eq:Ns}
	  N_s = 2\left\lceil\frac{4}{\alpha(\delta)^2}\left(\log\frac{4}{\eta}+\log\left(\left\lceil\log_2\frac{1}{\eps}\right\rceil+1\right)\right)\right\rceil, 
	\end{equation}
    where $\alpha(\delta)$ is defined in \eqref{eq:alpha}. 
	Then the output of \Cref{alg:pruning} satisfies
	\begin{equation}
		\PP\left(\ntp{\theta_J - \lambda_0} < \frac{\pi}{3}\eps\right) > 1-\eta.
	\end{equation}
	In addition, the maximal runtime and the total cost of \Cref{alg:pruning} are, respectively,
	\begin{equation}
		\tm = \mo{\eps^{-1}},\quad \tt = \mo{\eps^{-1}\left(\log(\eta^{-1})+\log\log(\eps^{-1})\right)}.
	\end{equation}
\end{theorem}

\begin{proof}
  First, we claim that under the circumstances of
  \begin{equation}
	\abs{Z_j - \braket{\psi|U^{2^j}|\psi}} < \alpha(\delta),\quad \forall j = 0,1,\ldots,J,\label{eq:cond}
  \end{equation}
  $\ntp{\theta_J - \lambda_0} < \frac{\pi}{3}\eps$. In fact, we can show
  \begin{equation}
	\lambda_0 \in \left(\theta_j-\frac{\pi}{3\cdot 2^j},  \theta_j+\frac{\pi}{3\cdot 2^j}\right) \mod 2\pi,
  \label{eq:int}
  \end{equation} 
  by induction. The case $j=0$ is a direct consequence of \Cref{lem:1}. Now if \eqref{eq:int} holds for $j-1$, then \Cref{lem:1} states that $\lambda_0$ is in one of the intervals $I_k = \left(\frac{2k\pi +\arg Z_j-\pi/3}{2^j}, \frac{2k\pi +\arg Z_j+\pi/3}{2^j}\right)$ for $k=0,\ldots,2^j-1$. Checking the length of the gaps between these intervals shows that only one such interval $I_{k_*}$ with $k_* = \arg\min_k \ntp{\frac{2k\pi +\arg Z_j}{2^j} - \theta_{j-1}}$ can have a non-empty intersection with $\left(\theta_{j-1}-\frac{\pi}{3\cdot 2^{j-1}}, \theta_{j-1}+\frac{\pi}{3\cdot 2^{j-1}}\right)$ (see \Cref{fig:thm}). Recall that this is exactly the criteria for choosing $\theta_j$. Hence $\lambda_0\in I_{k_*} = \left(\theta_j-\frac{\pi}{3\cdot 2^j}, \theta_j+\frac{\pi}{3\cdot 2^j}\right)$.
	
  It remains now to prove that \eqref{eq:cond} holds with a probability greater than $1-\eta$. By Hoeffding's inequality, using $N_s$ samples ensures that
  \begin{equation}
	\PP\left(\abs{Z_j - \braket{\psi|U^{2^j}|\psi}} < \alpha(\delta)\right) > 1-\frac{\eta}{J+1}
  \end{equation}
  for every $j$. Therefore, we conclude that
  \begin{equation}
	\PP\left(\abs{Z_j - \braket{\psi|U^{2^j}|\psi}} < \alpha(\delta) \text{ hold for every } j = 0,1,\ldots,J\right) > 1-\eta.
  \end{equation}
\end{proof}

\begin{figure}[ht!]
  \centering
  \includegraphics[scale=0.4]{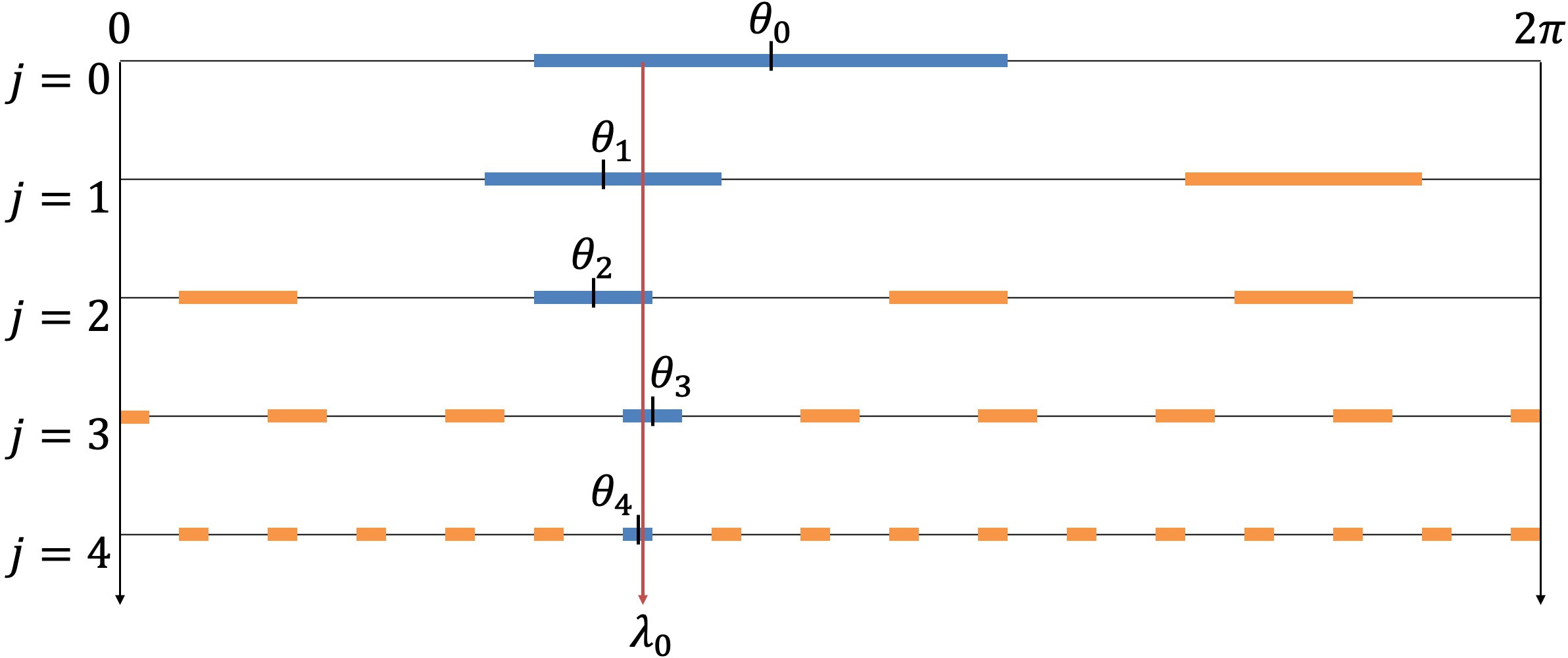} 
  \caption{Illustration of \Cref{alg:pruning} and \Cref{thm:1}. At each iteration $j$, the blue interval is the neighborhood of the chosen $\theta_j$ with length $\frac{\pi}{3\cdot2^{j-1}}$. The yellow ones are discarded.}
  \label{fig:thm}
\end{figure}

\section{Low-depth circuit for large initial overlap}\label{sec:new}

This section shows that the maximal runtime can be further reduced when the initial overlap $p_0$ is close to one (i.e., the upper bound $\delta$ of the noise in the initial state $\ll 1$). This is achieved via \Cref{alg:new}.

\begin{algorithm}[ht]
	\caption{Phase estimation in the regime of large overlap}
	\label{alg:new}
	\begin{algorithmic}
		\STATE{\textbf{Input:} $\eps$: target accuracy, $\eta$: upper bound of the failure probability,  $\delta$: upper bound for the noise in the initial state $\ket{\psi}$, $\xi$: additional prefactor for the maximal runtime that must satisfy $1>\xi>\frac{3}{\pi}\arcsin((1-\delta)^{-1}\delta)$. }
		\STATE{Let $J = \lceil\log_2(\frac{\xi}{\eps})\rceil$ and calculate $N_s$ according to the values of the inputs and \eqref{eq:Ns_new}.}
		\STATE{$\theta_{-1} = 0$.}
		\FOR{$j = 0,1,\ldots,J$}
			\STATE{Run the circuit in \Cref{fig:circuits}(b) for the real part and imaginary part of $\braket{\psi|U^{2^j}|\psi}$ for $\frac{N_s}{2}$ times each to generate $Z_j$ as an estimation of $\braket{\psi|U^{2^j}|\psi}$.}
			\STATE{Define a candidate set $S_j = \left\{\frac{2k\pi+\arg Z_j}{2^j}\right\}_{k=0,\ldots,2^j-1}$.}
			\STATE{$\theta_j = \arg\min_{\theta\in S_j}\ntp{\theta - \theta_{j-1}}$. }	
		\ENDFOR
		\STATE{\textbf{Output:} $\theta_J$ as an approximation to $\lambda_0$. }
	\end{algorithmic}
\end{algorithm}

In order to analyze the complexity of \Cref{alg:new}, we need a generalization of \Cref{lem:1} as provided below.
\begin{lemma}\label{lem:2}
Suppose that $\delta < 2\sqrt{3}-3$ and $1>\xi>\frac{3}{\pi}\arcsin((1-\delta)^{-1}\delta)$. Let 
	\begin{equation}
	  \beta(\delta, \xi) = (1-\delta)\sin\frac{\pi\xi}{3} - \delta >0.
	\end{equation} 
	If the quantum state $\ket{\psi}$ satisfies $p_0 > 1-\delta$ and an estimator $Z_j$ satisfies
	\begin{equation}
		\abs{Z_j - \braket{\psi|U^{2^j}|\psi}} < \beta(\delta, \xi),
	\end{equation} 
	then 
	\begin{equation}
		2^j\lambda_0 \in \left(\arg Z_j-\frac{\pi\xi}{3}, \arg Z_j+\frac{\pi\xi}{3}\right)\mod 2\pi.\label{eq:range_j_new}
	\end{equation}
\end{lemma}

\begin{proof}
	Similar with the proof of \Cref{lem:1}, we have
	\begin{equation}
		\begin{aligned}
		  \beta(\delta, \xi) &> \abs{Z_j - \braket{\psi|U^{2^j}|\psi}}	= \abs{Z_j - p_0 e^{\mathrm{i} 2^j \lambda_0} - \sum_{m = 1}^M p_m e^{\mathrm{i} 2^j \lambda_m}}\\
			&\ge \abs{Z_j - p_0 e^{\mathrm{i} 2^j \lambda_0}} - \sum_{m = 1}^M p_m  \ge \abs{Z_j - p_0 e^{\mathrm{i} 2^j \lambda_0}} - \delta.
		\end{aligned}
	\end{equation}
	Hence $Z_j$ is in a ball $B_{\beta(\delta, \xi)+\delta}(p_0 e^{\mathrm{i} 2^j \lambda_0}) = B_{(1-\delta)\sin\frac{\pi\xi}{3}}(p_0 e^{\mathrm{i} 2^j \lambda_0})$. Since $p_0>1-\delta$, the sine of the angle between $Z_j$ and $e^{\mathrm{i} 2^j \lambda_0}$ is upper-bounded by $\sin\frac{\pi\xi}{3}$ and thus \eqref{eq:range_j_new} holds.
\end{proof}

\begin{rmk}\label{rmk:smalldelta}
    If we know a priori that the overlap between $\ket{\psi}$ and $\ket{\psi_0}$ is large, i.e., $\delta \ll 1$, we have $\frac{3}{\pi}\arcsin((1-\delta)^{-1}\delta)=\Theta(\delta)$. Thus, the constraint on $\xi$ enforced in \Cref{lem:2} is $\xi=\Omega(\delta)$ in the regime of large overlap. 
\end{rmk}

Here, we generalize the result presented in \Cref{thm:1} and show that \Cref{alg:new} can give a further reduced $\tm$. Notice that even though the assumption on $\delta$ still reads $\delta < 2\sqrt{3}-3$, \Cref{thm:2} only provides substantial reduction on $\tm$ when $\delta$ is sufficiently small. 

\begin{theorem}\label{thm:2}
	Assume that the quantum state $\ket{\psi}$ satisfies $p_0 > 1-\delta $, where $\delta < 2\sqrt{3}-3$. Suppose that $1>\xi>\frac{3}{\pi}\arcsin((1-\delta)^{-1}\delta)$, and that the parameter $N_s$ in \Cref{alg:new} is given by: 
	\begin{equation}\label{eq:Ns_new}
	  N_s = 2\left\lceil\frac{4}{\beta(\delta,\xi)^2}\left(\log\frac{4}{\eta}+\log\left(\left\lceil\log_2\frac{\xi}{\eps}\right\rceil+1\right)\right)\right\rceil.
	\end{equation}
	Then the output of \Cref{alg:new} satisfies
	\begin{equation}
		\PP\left(\ntp{\theta_J - \lambda_0} < \frac{\pi}{3}\eps\right) > 1-\eta.
	\end{equation}
	In addition, the maximal runtime and the total cost of \Cref{alg:new} are, respectively,
	\begin{equation}
		\tm = \mo{\xi\eps^{-1}},\quad \tt = \mo{\frac{\xi\eps^{-1}}{\beta(\delta, \xi)^2}\left(\log(\eta^{-1})+\log\log(\eps^{-1})\right)}.
	\end{equation}
\end{theorem}

\begin{figure}[ht!]
  \centering
  \includegraphics[scale=0.4]{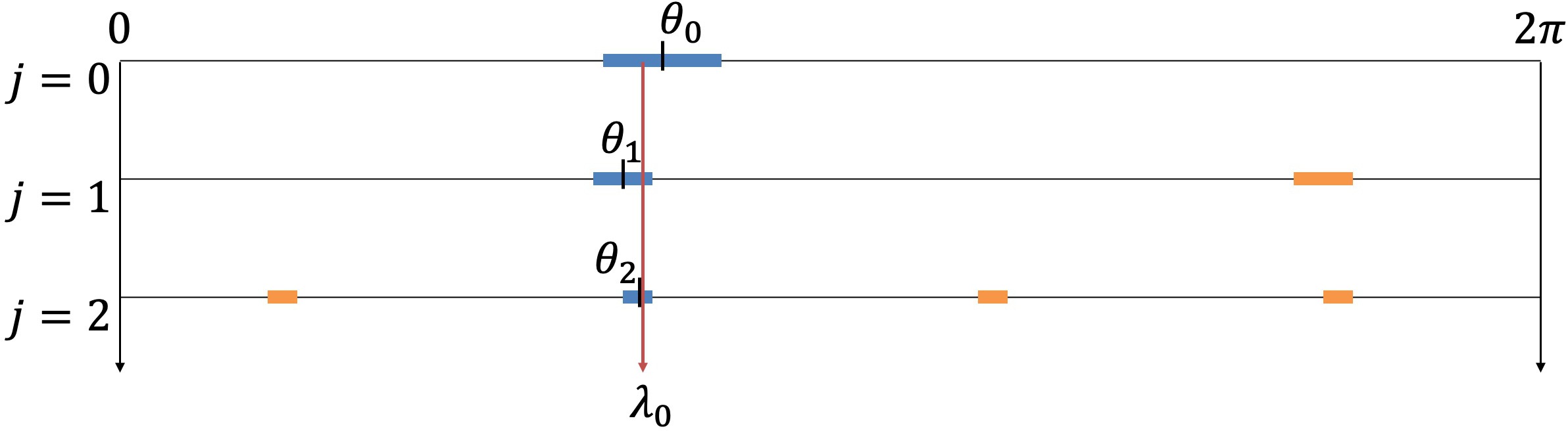} 
  \caption{Illustration of \Cref{alg:new} and \Cref{thm:2}. Compared to \Cref{alg:pruning}, a short interval is maintained at each level, and thus, the maximum value $J$ is smaller even for the same precision $\eps$, leading to a smaller $\tm$.}
  \label{fig:thmnew}
\end{figure}

\begin{proof}
  Similar to \Cref{thm:1}, we first prove that under the circumstance of
  \begin{equation}
	\abs{Z_j - \braket{\psi|U^{2^j}|\psi}} < \beta(\delta, \xi),\quad \forall j = 0,1,\ldots,J,\label{eq:cond_new}
  \end{equation}
  $\ntp{\theta_J - \lambda_0} < \frac{\pi}{3}\eps$. We proceed to show
  \begin{equation}
	\lambda_0 \in \left(\theta_j-\frac{\pi\xi}{3\cdot 2^j},  \theta_j+\frac{\pi\xi}{3\cdot 2^j}\right) \mod 2\pi, \label{eq:int_new}
  \end{equation} 
  by induction. The case $j=0$ is a direct corollary of \Cref{lem:2}. Now assume that \eqref{eq:int_new} holds for $j-1$. It is guaranteed by \Cref{lem:2} that $\lambda_0$ is in one of the intervals $I_k = \left(\frac{2k\pi+\arg Z_j-\pi\xi/3}{2^j}, \frac{2k\pi+\arg Z_j+\pi\xi/3}{2^j}\right)$ for $k=0,\ldots,2^j-1$. Since $\xi<1$, the same argument as in \Cref{thm:1} shows that at most one such interval $I_{k_*}$ with $k_* = \arg\min_k \ntp{\frac{2k\pi+\arg Z_j}{2^j} - \theta_{j-1}}$ can have a non-empty intersection with $\left(\theta_{j-1}-\frac{\pi\xi}{3\cdot 2^{j-1}}, \theta_{j-1}+\frac{\pi\xi}{3\cdot 2^{j-1}}\right)$ (see \Cref{fig:thmnew}). Since this is the criteria for choosing $\theta_j$ in the algorithm, we have $\lambda_0\in I_{k_*} = \left(\theta_j-\frac{\pi\xi}{3\cdot 2^j}, \theta_j+\frac{\pi\xi}{3\cdot 2^j}\right)$. Plugging in $J = \lceil\log_2(\frac{\xi}{\eps})\rceil$ yields $\ntp{\theta_J - \lambda_0} < \frac{\pi}{3}\eps$.

  Now we show that \eqref{eq:cond_new} holds with a probability greater than $1-\eta$. In fact, by Hoeffding's inequality, with $N_s$ samples from the Hadamard test, one can ensure that
  \begin{equation}
	\PP\left(\abs{Z_j - \braket{\psi|U^{2^j}|\psi}} < \beta(\delta, \xi)\right) > 1-\frac{\eta}{J+1}
  \end{equation}
  for every $j$. Thus, with the union bound, one arrives at
  \begin{equation}
	\PP\left(\abs{Z_j - \braket{\psi|U^{2^j}|\psi}} < \beta(\delta, \xi) \text{ hold for every } j = 0,1,\ldots,J\right) > 1-\eta.
  \end{equation}
\end{proof}

\begin{rmk}
  It is worth noticing that when reducing $\tm$ to $\mo{\xi\eps^{-1}}$, $\tt$ will increase to $\tmo{\xi^{-1}\eps^{-1}}$. Therefore, a trade-off similar to \cite{ding2022even} exists between $\tt$ and $\tm$. In particular, by using a small prefactor $\xi=\Theta(\delta)$ in the large overlap regime, one reduces $\tm$ to $\mo{\delta\eps^{-1}}$ (while increasing the total runtime to $\tmo{\delta^{-1}\eps^{-1}}$). According to \Cref{rmk:smalldelta}, \Cref{alg:new} reduces the maximal circuit length by a factor of $\Theta(1/\delta)$, while the result presented in \cite{ding2022even} only gives a factor of $\Theta(1/\sqrt{\delta})$. After the initial submission of our manuscript, an extended version of QCELS is shown in \cite{ding2023simultaneous} to achieve the factor $\Theta(1/{\delta})$ with a more sophisticated analysis. 
\end{rmk}

\begin{rmk}
The result in \Cref{alg:new} can be extended to the case with large relative overlap (Cf. \cite{lin2022heisenberg, ding2022even}) without increasing the depth of the quantum circuit. The filtering process and proof in \cite{ding2022even} can be directly applied here. Briefly speaking, if one knows a priori that $\lambda_0\in I\subset I'\subset [-\pi, \pi]$, then one only needs $\frac{p_0}{\sum_{\lambda\in I'}p_m}>1-\delta$ instead of $p_0>1-\delta$. The idea is to replace $Z_j$, which is an estimator of $\braket{\psi|U^{2^j}|\psi}=\sum_{m=0}^{M-1}p_m2^j\lambda_m$, with an estimator of 
\[
\braket{\psi|U^{2^j}F_q(H)|\psi}\approx\braket{\psi|U^{2^j}\mathbf{1}_{I'}(H)|\psi}=\sum_{\lambda_m\in I'}p_m2^j\lambda_m.
\]
Here $H$ is a matrix such that $U=e^{\mathrm{i} H}$ and $F_q$ is an approximation of the indicator $\mathbf{1}_{I'}$ up to precision $q$, which is obtained by a truncation of the Fourier series of $\mathbf{1}_{I'}$. The estimator of $\braket{\psi|U^{2^j}F_q(H)|\psi}$ is then constructed with the results of Hadamard tests and the coefficients of the Fourier series. 
\end{rmk}

{
\begin{figure}[ht!]
  \centering
  \begin{tabular}{cc}
    \includegraphics[width =0.45\textwidth]{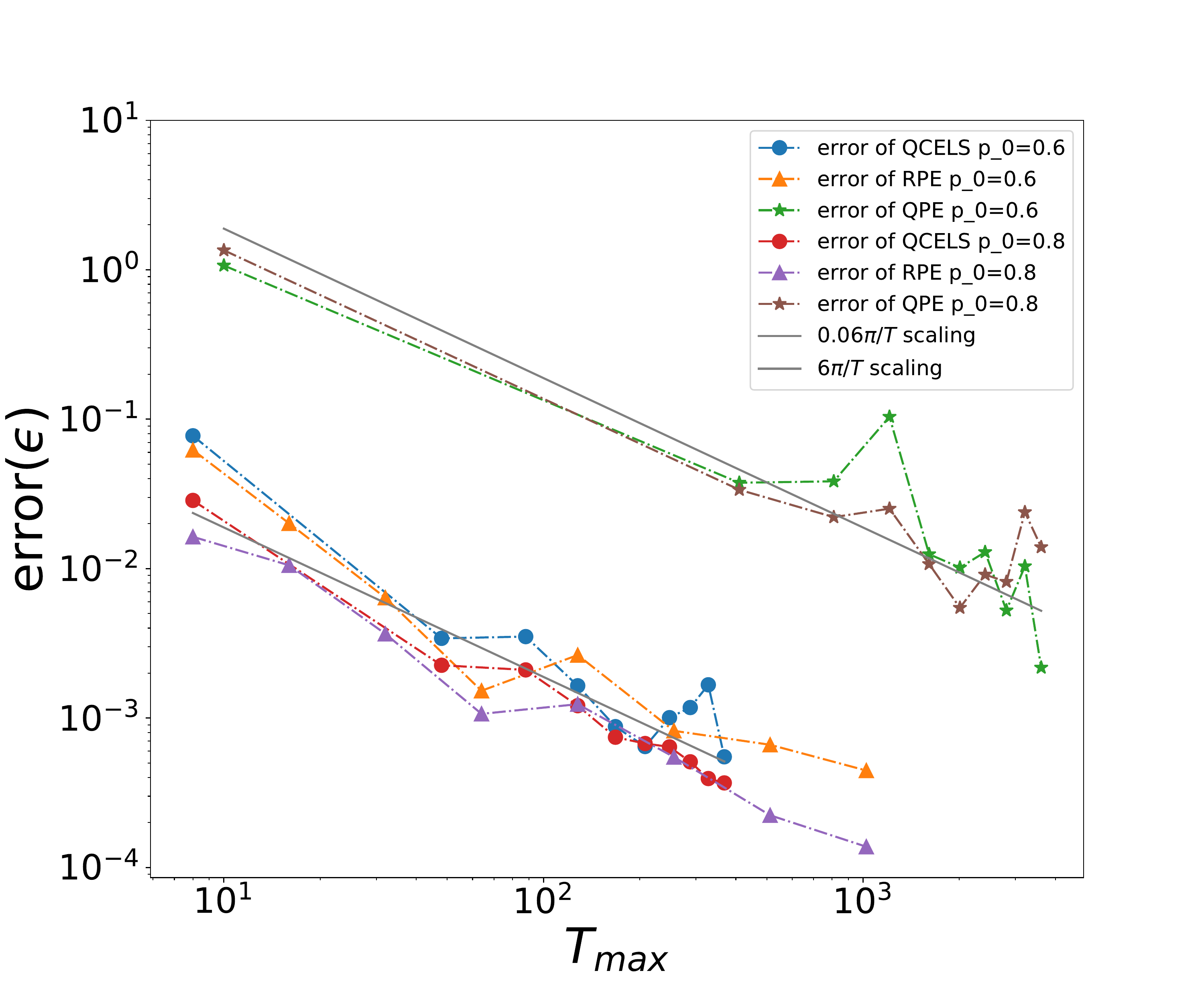}&
    \includegraphics[width =0.45\textwidth]{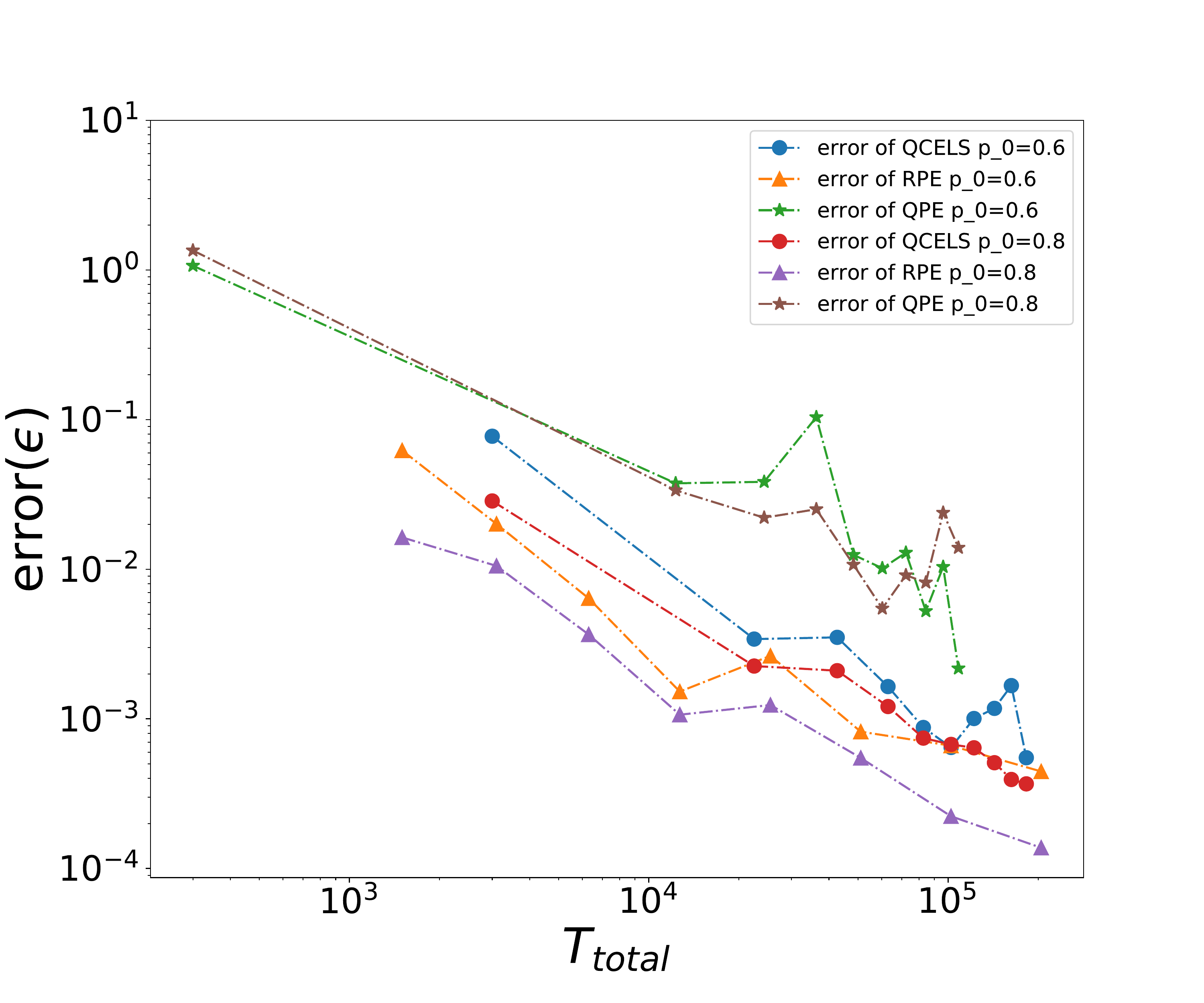}\\
    (a) & (b)
  \end{tabular}
  \caption{Numerical simulations for the transverse field Ising model using different phase estimation methods. RPE, QCELS, and QPE denote our method, the optimization-based method in \cite{ding2022even}, and textbook version QPE, respectively. The initial overlap $p_0$ is chosen to be 0.6 and 0.8. (a) Comparison of the maximal runtime $\tm$. (b) Comparison of the total runtime $\tt$. }
  \label{fig:compare}
\end{figure}

\section{Numerical simulation}

In this section, we present the results of numerical simulations of our algorithms and compare our method to other phase estimation algorithms. The unitary operator we employ is defined as $U = e^{\mathrm{i}\frac{\pi}{4} H/\|H\|_2}$, where $H$ is a Hamiltonian and $\|\cdot\|_2$ represents the operator $2$-norm. The scaling factor $\frac{\pi}{4\|H\|_2}$ in the exponential is used to ensure that the eigenvalues are all in $[-\frac{\pi}{4},\frac{\pi}{4}]$, thereby eliminating any ambiguity arising from modulo $2\pi$. 

Specifically, we employ the one-dimensional transverse field Ising model with $L$ sites and periodic boundary conditions as the Hamiltonian, which is given by} 
{
\[
H=-\left(\sum_{i=1}^{L-1} Z_i Z_{i+1}+Z_L Z_1\right)-g \sum_{i=1}^L X_i,
\]
with parameters $L=8$ and $g=4$. Here, $X_i$ and $Z_i$ represent the Pauli matrices acting on the $i$-th site.

In \Cref{fig:compare}, we present the performance of our RPE algorithm (\Cref{alg:pruning}) and compare it to two other algorithms, namely the QCELS algorithm (the optimization-based method in \cite{ding2022even}) and the textbook version QPE \cite{cleve1998quantum}, with two different initial overlap. The parameters for QCELS are set identically to those in \cite{ding2022even}. All data for these three methods are obtained by conducting ten random experiments and calculating the average error. As demonstrated in \Cref{fig:compare}(a), the error given by RPE and QCELS decreases when increasing $\tm$ with a linear trend for both $p_0=0.6$ and $p_0=0.8$, and both RPE and QCELS provide a much smaller prefactor than textbook QPE. On the other hand, it is clear from \Cref{fig:compare}(b) that \Cref{alg:pruning} achieves a lower total cost $\tt$ compared to QCELS. 

\begin{figure}[ht!]
  \centering
  \begin{tabular}{cc}
    \includegraphics[width =0.45\textwidth]{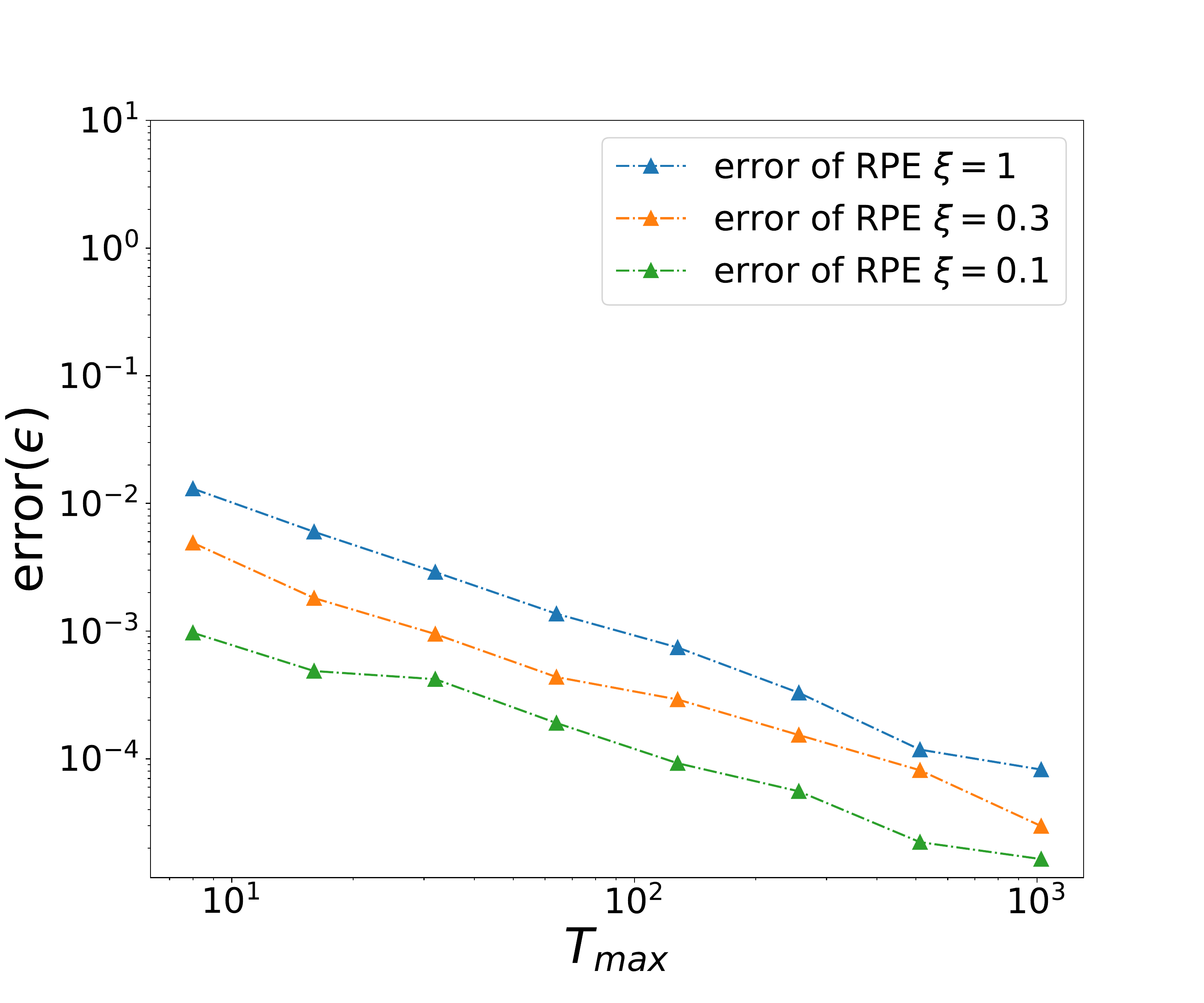} &
    \includegraphics[width =0.45\textwidth]{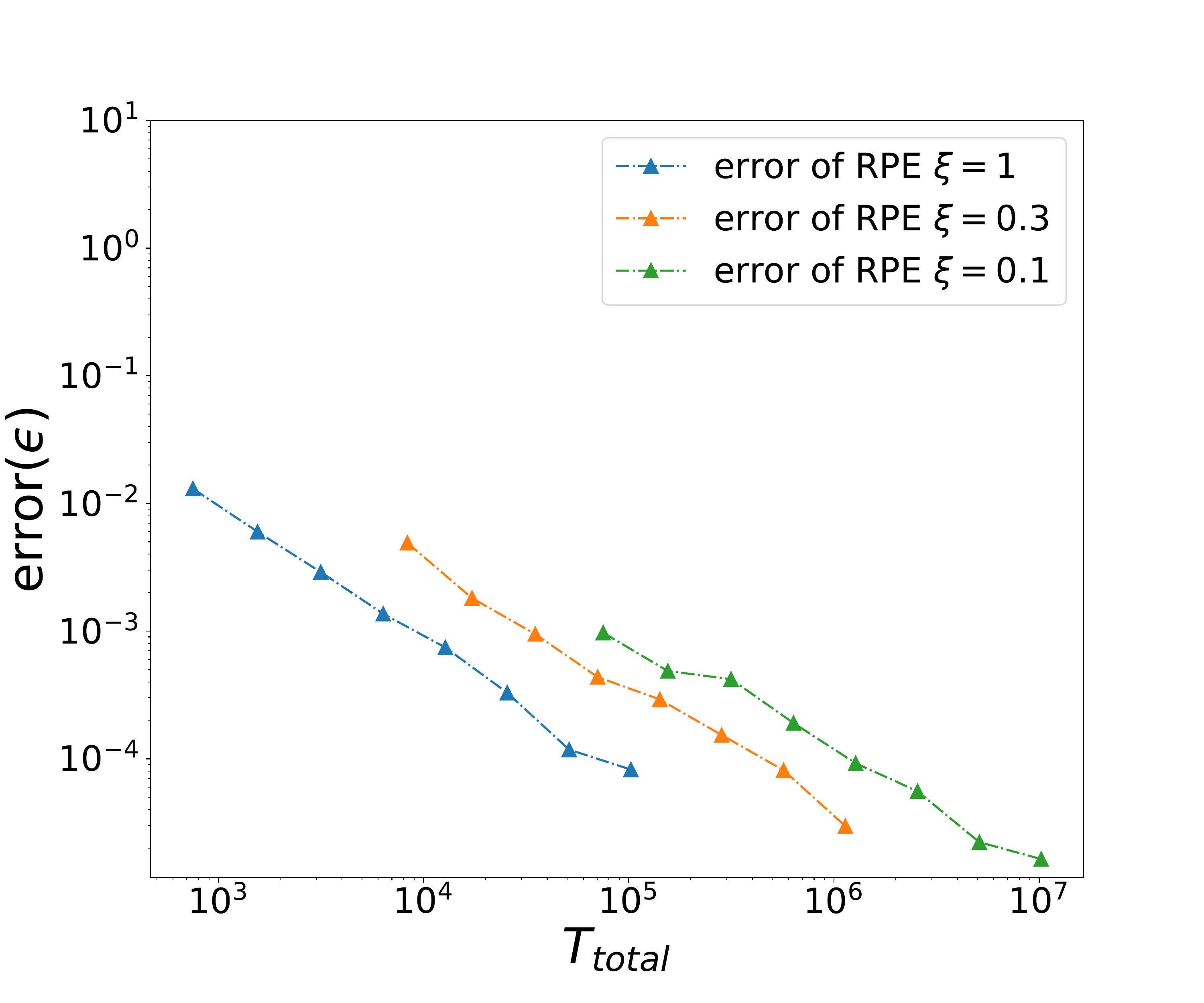}\\
    (a) & (b)
  \end{tabular}
  \caption{Numerical simulation for transverse field Ising model with different $\xi$'s. The initial overlap $p_0$ is chosen to be 0.99. The logarithmic scale is used for both the vertical and the horizontal axes. (a) Comparison of maximal runtime. (b) Comparison of total runtime.  }
  \label{fig:diff_xi}
\end{figure}

In the second numerical experiment, we consider the same unitary operator $U$ but assume that the initial overlap is sufficiently large, and we proceed to verify that \Cref{alg:new} can further reduce the prefactor in $\tm$ by lowering the value of $\xi$ in this regime. In particular, we assume that the initial overlap is $p_0=0.99$. Then $\xi$ can be as small as $\frac{3}{\pi}\arcsin(0.01/0.99)\approx0.01$ according to \Cref{thm:2}. \Cref{fig:diff_xi} displays the results for \Cref{alg:new} with different values of $\xi$, where the error is also averaged from $10$ random experiments. For the three different values $\xi = 1$, $0.3$, and $0.1$, the error shows a similar linear decreasing trend while the maximal runtime and the total runtime increase, which indicates that the difference only lies in the prefactor. Moreover, the prefactor of $\tm$ decreases as $\xi$ decreases, at the expense of an increase in $\tt$, which verifies the conclusions of \Cref{thm:2}. 
}

\section{Conclusion and discussion}
In this paper, we have demonstrated through theoretical analyses and numerical experiments that a simple RPE-type algorithm and its variant can be particularly suitable for the implementation of phase estimation on early fault-tolerant quantum computers since they satisfy the requirements (\ref{1}), (\ref{2}), (\ref{3}) and (\ref{4}). Compared with the previous work \cite{ding2022even}, our method is structurally simpler since no optimization procedure is needed. Our method also provides a more significant prefactor reduction while at the same time posing a looser requirement for the initial overlap. 

For the theoretical results presented in this paper, a large probability statement is adopted. One can also easily extend the results to the other metrics, such as the mean squared error (MSE), where a different number of measurements can be used in each iteration to minimize the MSE. 

The unitary $U$ is assumed to be a black-box unitary in the setting of this paper. When other powers of $U$ apart from $U^{2^j}$ are accessible, it is possible to relax further the requirement $\delta < 2\sqrt{3}-3$ (see the follow-up work \cite{li2023low} for details). The fact that $\delta$ only needs to be smaller than an $\mo{1}$ threshold makes the algorithm presented here specifically suitable for the case where the unitary comes from a simulation of a Hamiltonian of interest. In that case, the result shown here only requires an approximate simulation with $\mo{1}$ precision for each time duration $2^j$, thus making the algorithm particularly advantageous for the combination with Hamiltonian simulation algorithms. 

\bibliographystyle{abbrvnat}
\bibliography{ref}

\end{document}